%% file: main.tex
\documentclass[11pt]{article}

\usepackage[utf8]{inputenc}
\usepackage[T1]{fontenc}
\usepackage{amsmath,amssymb,amsthm}
\usepackage{booktabs}
\usepackage{float}
\usepackage{graphicx}
\usepackage{geometry}
\usepackage[hidelinks]{hyperref}
\usepackage[numbers,sort&compress]{natbib}
\usepackage{placeins}
\usepackage{xcolor}

\graphicspath{{../results/figures/}{./}}
\input{experiment_macros}

\newtheorem{definition}{Definition}
\newtheorem{proposition}{Proposition}

\title{Epistemic Horizon Minority Games:\\
When Abundance Reduces Strategic Value}
\author{Faruk Alpay\thanks{Corresponding author: \texttt{alpay@lightcap.ai}} \quad Levent Sar{\i}o\u{g}lu\\[3pt]
\small Department of Computer Engineering, Bah\c{c}e\c{s}ehir University, Istanbul, T\"urkiye\\[-1pt]
\small \texttt{\{faruk.alpay, levent.sarioglu\}@bahcesehir.edu.tr}}
\date{}

\begin{document}
\maketitle

\begin{abstract}
Strategic value can fall when an option becomes visible. A route, signal, bet,
or opportunity may be attractive because few agents see it; public attention can
erase the advantage it reveals. We formalize this mechanism as an
epistemic-horizon minority game (EHMG), where agents have bounded observation
horizons, action-specific awareness, desire-biased utilities, and payoffs that
decline with crowding. The object is not a fixed congestion game with omitted
actions, but an awareness-transition game on a finite lattice. We prove
fixed-awareness potential-game reduction, finite
monotone awareness convergence, logit mean-field uniqueness under an explicit
norm condition, non-reducibility from static count-based congestion games, and
sensitivity bounds for nonlinear revelation. We separate the target price of
information from aggregate welfare loss, showing that they can coincide,
diverge, or recommend opposite disclosure policies. Private revelation, public
common revelation, and correlated group disclosure are modeled as distinct
signal structures with different equilibrium effects. Experiments regenerate
awareness sweeps, public visibility shocks, horizon--desire grids,
information-constrained Braess examples, disclosure optimization, minimum
harmful revelation, and counterfactual baselines isolating the epistemic
mechanism from ordinary full-awareness congestion. Strategic trace encodings
are evaluated as a controlled regime-recognition benchmark: raw trajectories,
Fourier summaries, recurrence and Gramian images, image bundles, local-filter
features, leakage probes, phase-scrambled controls, resolution and
recurrence-threshold sweeps, spectral carriers, and IAAFT matched-null
parameter-shift controls test whether trace encodings recover strategic
structure under robust nulls.
\end{abstract}

\section{Introduction}

In ordinary language, abundance suggests value. In strategic systems this
intuition often fails. If a bar is enjoyable only when few people attend, a
forecast that the bar will be empty can make it crowded. If a route is faster
only when few drivers take it, public route guidance can erase its advantage. If
a bet is attractive because a player wants a particular hand to win, the
player's own desire can leak into the perceived probability of success. In each
case, value is not attached to the object alone. It is attached to how many
agents see the object, how many desire it, and what each agent believes others
will do.

The El Farol bar problem was introduced by Arthur as a compact model of
inductive reasoning under bounded rationality \citep{arthur1994inductive,
arthur2022elfarol}. Challet and Zhang's minority game placed the same intuition
into a game form in which agents benefit from being on the less crowded side
\citep{challet1997minority,yeung2009minority}. Congestion games formalize the
general case in which an action's payoff is a decreasing function of the number
of users \citep{rosenthal1973congestion}. Braess' paradox shows that adding an
option to a decentralized network can make everyone worse off
\citep{braess1968paradoxon,nagurney2020braess}; informational Braess shows that
additional information about options can also be harmful
\citep{acemoglu2018informational}. These literatures capture one part of the
intuition: more access can reduce value.

The missing part is epistemic. Agents do not choose in the world as a whole;
they choose inside an observation horizon. A person walking through a forest can
be locally coherent while omitting most simultaneous events from the decision
model. A gambler can know that an event is random while still shifting a
borderline prediction toward the wanted outcome. A community can disagree about
an unobservable domain not merely because agents assign different probabilities,
but because their state spaces do not contain the same objects. Epistemic game
theory makes beliefs, higher-order beliefs, common knowledge, and unawareness
formal objects \citep{aumann1976agreeing,pacuit2025epistemic,
vanderschraaf2022common,heifetz2006unawareness}. Behavioral decision theory
adds the empirical fact that desire and risk framing can distort judgment
\citep{kahneman1979prospect,tversky1971smallnumbers,windschitl2010desirability,
croson2005gambler}.

This paper combines these two sides. We introduce an \emph{epistemic horizon
minority game} (EHMG), a repeated negative-frequency game in which each agent
has a bounded memory horizon and a possibly incomplete action set. An action can
be physically available but absent from an agent's awareness. When the action is
revealed, it may gain demand because it is newly visible, but lose payoff
because its prior advantage was epistemic scarcity.

\paragraph{Contributions.}
\begin{itemize}
  \item We define EHMGs with awareness sets, observation horizons,
  desirability-biased perceived utilities, and realized payoffs that decline
  with action crowding.
  \item We introduce an awareness lattice and transition operators for private,
  public, and correlated revelation; this yields a mean-field awareness
  equilibrium object and finite convergence result for monotone awareness
  expansion.
  \item We prove a theorem family for visibility: fixed-awareness EHMGs reduce
  to potential games, but awareness transitions create non-reducible comparative
  statics; under logit contraction the mean-field equilibrium is unique and its
  price of information is bounded for nonlinear congestion.
  \item We embed the Braess comparison inside EHMG as an
  information-constrained Wardrop instance rather than using it only as an
  analogy.
  \item We define algorithmic disclosure tasks, including budgeted
  welfare-safe disclosure and minimum harmful revelation, and give exact,
  greedy, oracle, and treewidth-oriented computational baselines.
  \item We materialize counterfactual experiments against ordinary congestion,
  full-awareness minority dynamics, partial revelation, and no-scarcity-bonus
  controls.
  \item We adapt a reusable algorithmic layer for strategic-state oracle
  bounds, compatible-coalition hardness witnesses, bounded-dependence mechanism
  kernelization, policy surprise, reduced visibility dynamics, and binary
  belief repair.
  \item We evaluate trace-image encodings, phase controls, parameter-shift
  splits, IAAFT marginal/spectrum-matched nulls, raw, Fourier,
  recurrence-image, ROCKET-style time-series, and lightweight local-filter
  baselines, plus an external-trace adapter.
\end{itemize}

\section{Related Work}

\paragraph{El Farol, minority games, and congestion.}
Arthur's El Farol problem uses a simple attendance threshold to show why a
single deductive forecast can self-negate \citep{arthur1994inductive}. The
minority game abstracts this into a repeated game in which success comes from
being in the smaller group \citep{challet1997minority,yeung2009minority}.
Rosenthal's congestion games show that finite games with resource congestion
have pure Nash equilibria via a potential function
\citep{rosenthal1973congestion}. EHMGs inherit this negative-frequency logic but
make action awareness and observation horizon explicit state variables.

\paragraph{More options, worse outcomes.}
Braess' paradox is the canonical warning that adding capacity or route choice can
harm decentralized users \citep{braess1968paradoxon,nagurney2020braess}. The
informational Braess paradox is closer to the present paper: users with expanded
information sets can experience higher equilibrium costs
\citep{acemoglu2018informational}. Our model is not a replacement for network
Wardrop analysis. It is a compact agent-based counterpart for settings where
the newly visible object is not just a road, but an action, opportunity, or
belief target.

\paragraph{Bounded reasoning and desire.}
Human play often departs from exact equilibrium. Cognitive hierarchy models
allow agents to reason at finite depths \citep{camerer2004cognitive}; quantal
response equilibrium replaces sharp best responses with noisy probabilistic
responses \citep{mckelvey1995qre,goeree2018qre}. Prospect theory, the law of
small numbers, gambler's fallacy, and desirability bias show why a wanted
outcome can affect probability judgment and prediction thresholds
\citep{kahneman1979prospect,tversky1971smallnumbers,windschitl2010desirability,
croson2005gambler}. EHMG uses these ideas modestly: desire is an additive
preference term in perceived utility, not a claim that all probability
distortion has one cause.

\paragraph{Epistemic limits.}
Epistemic models of games study how rationality, belief, and higher-order
belief support solution concepts \citep{pacuit2025epistemic}. Common knowledge is
central to coordination \citep{vanderschraaf2022common}, while Aumann's theorem
shows how strong common-prior and common-knowledge assumptions constrain
disagreement \citep{aumann1976agreeing}. Unawareness models go further:
players may fail to represent actions, players, or states at all
\citep{heifetz2006unawareness}. EHMG uses this distinction to separate
uncertainty about an outcome from absence of the outcome from an agent's model.

\paragraph{Visual encodings of dynamics.}
Time-series imaging maps one-dimensional traces into spatial objects that can be
read by visual inspection or image-analysis pipelines. Gramian angular fields
and related maps were introduced for time-series classification and imputation
\citep{wang2015imaging}; recurrence plots visualize repeated states in a
dynamical system \citep{eckmann1987recurrence}; and Fourier phase is a classical
control variable for local image structure \citep{oppenheim1981phase}. We use
these tools conservatively: they are diagnostics for strategic trajectories, not
a substitute for the formal model.

\paragraph{Disclosure optimization.}
Once visibility becomes a decision variable, the model touches standard
algorithmic questions: choose which information to reveal, under what budget,
and with what welfare or harm guarantees. The minimum harmful revelation task
below is a covering problem in disguise, so the usual set-cover hardness
phenomena are relevant \citep{karp1972reducibility,feige1998setcover}. This
connection makes the algorithmic layer part of the EHMG core rather than a
separate collection of complexity proxies.

\section{Model}

Let \(N=\{1,\ldots,n\}\) be agents and \(A=\{1,\ldots,K\}\) actions. At round
\(t\), agent \(i\) is aware of a nonempty subset \(A_i(t)\subseteq A\). A public
visibility shock is an update \(A_i(t+1)\supseteq A_i(t)\) for a set of agents.
We distinguish three information structures. A \emph{private expansion} reveals
an action to selected agents without making the event common knowledge. A
\emph{public revelation} makes the action visible to all agents and makes that
visibility itself public. A \emph{correlated disclosure} reveals to a subset
whose membership is statistically related, for example through a platform,
community, or recommendation channel. These structures can have different
strategic effects even when the number of newly aware agents is the same.
Formally, agent \(i\)'s perceived awareness share for action \(a\) is
\(\widehat\pi_{i,a}\), a belief about how many other agents have \(a\) in their
awareness sets. Private revelation changes \(A_i\) but need not change
\(\widehat\pi_{j,a}\) for \(j\ne i\). Public revelation sets
\(\widehat\pi_{i,a}=1\) for all agents and makes that event common knowledge.
Correlated revelation gives agents a signal cell \(C_i\) and sets
\(\widehat\pi_{i,a}=\mathbb{E}[|S|/n\mid C_i]\). Thus ``I know \(a\)'',
``I know others know \(a\)'', and ``it is common knowledge that \(a\) is
available'' are different EHMG states.

\begin{proposition}[Correlated disclosure can differ at fixed reach]
Fix a focal action \(r\) and a recipient set \(S\) of size \(q\). Compare two
disclosures that reveal \(r\) to exactly the same agents in \(S\). In the first,
the recipient set is publicly announced, so each recipient uses
\(\widehat\pi^P_{i,r}=q/n\). In the second, recipients observe only a correlated
group signal \(C_i\) and use posterior \(\widehat\pi^C_{i,r}=p\). If
\(\eta>0\), \(p\ne q/n\), and all other perceived-utility terms are held fixed,
then the logit odds of choosing \(r\) differ by
\[
  \frac{\Pr_C(a_i=r)/\Pr_C(a_i=b)}
       {\Pr_P(a_i=r)/\Pr_P(a_i=b)}
  =
  \exp\!\left(\frac{\eta(q/n-p)}{\tau}\right)
\]
against any common alternative \(b\in A_i\). Hence two disclosures with the same
number of newly aware agents can induce different adoption, target
price-of-information, and welfare effects.
\end{proposition}

\begin{proof}
The two disclosures produce the same feasible action set for every agent, so
they differ only through the awareness-share term in perceived utility. The
correlated signal changes \(\widehat u_{i,r}\) relative to the public
announcement by \(\eta(q/n-p)\), while \(\widehat u_{i,b}\) is unchanged. The
displayed odds ratio is the standard logit odds identity. If \(p<q/n\), the
correlated signal makes \(r\) appear more scarce and raises its adoption odds;
if \(p>q/n\), it lowers them. Because \(r\)'s realized payoff and welfare
contribution depend on adoption, the induced equilibrium effects can differ
even though the set of newly aware agents is identical.
\end{proof}

Let \(a_i(t)\in A_i(t)\) be the chosen action, and let
\[
  n_a(t)=\sum_{i=1}^n \mathbf{1}\{a_i(t)=a\}
\]
be the crowd on action \(a\). The realized payoff to agent \(i\) is
\[
  u_i(t)=v_{a_i(t)}-c\frac{n_{a_i(t)}(t)}{n},
\]
where \(v_a\) is intrinsic value and \(c>0\) is the congestion coefficient.
The experiments also use the more general congestion family
\[
  u_i(t)=v_{a_i(t)}-c\left(\frac{n_{a_i(t)}(t)}{n}\right)^\rho,
  \qquad \rho>0,
\]
so the main mechanism is not tied to a linear cost curve.

Agents do not observe the full process. Each has a horizon \(h_i\). Let
\(\widehat n_{i,a}(t)\) be the mean realized count of action \(a\) in the last
\(h_i\) rounds. The perceived utility used for choice is
\[
  \widehat u_{i,a}(t)
  =
  v_a
  -c\frac{\widehat n_{i,a}(t)}{n}
  +\eta \bigl(1-\pi_a(t)\bigr)
  +\lambda_i d_{i,a},
  \qquad a\in A_i(t),
\]
where \(\pi_a(t)=n^{-1}\sum_i \mathbf{1}\{a\in A_i(t)\}\) is realized awareness
share. In the information-structure version, agent \(i\) replaces this term by
the perceived share \(\widehat\pi_{i,a}(t)\). The parameter
\(\eta\ge 0\) is an epistemic-scarcity bonus, and \(\lambda_i d_{i,a}\) is the
desire term. The desire term is reduced-form: it can represent a utility tilt,
attention toward an action, or a prediction-threshold shift, but it is not meant
to identify gambler's fallacy, optimism, arousal, and desirability bias as one
mechanism. Actions outside \(A_i(t)\) have perceived utility \(-\infty\).
Choices follow a quantal response:
\[
  \Pr(a_i(t)=a)
  =
  \frac{\exp(\widehat u_{i,a}(t)/\tau)}
       {\sum_{b\in A_i(t)}\exp(\widehat u_{i,b}(t)/\tau)}.
\]

\begin{definition}[Epistemic scarcity]
An action \(a\) has epistemic scarcity at time \(t\) when \(\pi_a(t)<1\). Its
scarcity is not physical unavailability, but absence from some agents' action
spaces.
\end{definition}

\begin{definition}[Inverse abundance]
An action exhibits inverse abundance on an interval if increasing its awareness
or adoption rate decreases its realized payoff.
\end{definition}

\subsection{Awareness Transitions}

Let
\[
  \mathcal{L}=\prod_{i=1}^n 2^A
\]
be the finite awareness lattice ordered by componentwise inclusion. An awareness
profile \(\mathcal{A}=(A_1,\ldots,A_n)\in\mathcal{L}\) records which actions are
conceivable to each agent. For a focal action \(r\), a private revelation to
agent set \(S\subseteq N\) is the join operator
\[
  R^r_S(\mathcal{A})_i =
  \begin{cases}
    A_i\cup\{r\}, & i\in S,\\
    A_i, & i\notin S.
  \end{cases}
\]
Public revelation is \(R^r_N\), interpreted as both universal visibility and
common observability of the revelation event. Correlated revelation is
\(R^r_S\) where \(S\) is drawn from, or constrained by, a platform/community
sigma-field. This separates the number of newly aware agents from the
information structure by which they become aware.

For fixed \(\mathcal{A}\), let \(F_\tau(x;\mathcal{A})\) be the population
logit response map induced by perceived utilities and congestion share \(x\).
A \emph{mean-field awareness equilibrium} is a pair
\((x^\star,\mathcal{A}^\star)\) such that
\[
  x^\star = F_\tau(x^\star;\mathcal{A}^\star),
  \qquad
  \Gamma(\mathcal{A}^\star,x^\star)=\mathcal{A}^\star,
\]
where \(\Gamma:\mathcal{L}\times\Delta(A)\to\mathcal{L}\) is an awareness
transition operator.

\begin{proposition}[Finite convergence of monotone awareness expansion]
Suppose \(\Gamma\) is inflationary:
\(\Gamma(\mathcal{A},x)\succeq\mathcal{A}\) for every awareness profile and
population state. Then every awareness trajectory
\(\mathcal{A}_{t+1}=\Gamma(\mathcal{A}_t,x_t)\) reaches an awareness fixed point
after at most
\[
  nK-\sum_{i=1}^n |A_i(0)|
\]
strict expansion events, regardless of the intermediate population states
\(x_t\).
\end{proposition}

\begin{proof}
The lattice has one Boolean coordinate for each agent-action pair. Inflationary
updates can flip a coordinate only from absent to present and never back. The
displayed quantity is exactly the number of initially absent coordinates. Once
no coordinate flips, \(\Gamma\) has reached an awareness fixed point.
\end{proof}

\begin{proposition}[Existence of mean-field awareness equilibrium]
Suppose \(\Gamma\) is inflationary and eventually stationary on every monotone
awareness path. For every terminal awareness profile \(\mathcal{A}\), if
\(F_\tau(\cdot;\mathcal{A})\) is continuous on the action simplex, then there
exists a mean-field awareness equilibrium
\((x^\star,\mathcal{A}^\star)\).
\end{proposition}

\begin{proof}
By finite convergence, every monotone awareness path reaches some terminal
profile \(\mathcal{A}^\star\) with
\(\Gamma(\mathcal{A}^\star,x)=\mathcal{A}^\star\) for the relevant terminal
states. The simplex is compact and convex, and \(F_\tau\) is continuous because
logit probabilities are continuous in utilities. Brouwer's theorem gives
\(x^\star=F_\tau(x^\star;\mathcal{A}^\star)\).
\end{proof}

\begin{proposition}[Uniqueness under masked-logit contraction]
Fix an awareness profile and a finite collection of population types
\(\theta\), each with mass \(w_\theta\), awareness mask \(M_\theta\), and
masked-logit choice vector \(p_\theta(x)\). Let
\[
  F_\tau(x;\mathcal{A})=\sum_\theta w_\theta p_\theta(x)
\]
be the aggregate response map, where congestion enters utility through
\(-g_a(x_a)\). For each \(x\), let \(J_\theta(x)\) be the Jacobian of the
masked softmax with respect to feasible-action utilities and let
\(D(x)=\mathrm{diag}(g'_1(x_1),\ldots,g'_K(x_K))\), with unavailable
coordinates zeroed by \(M_\theta\). If
\[
  L=\sup_{x\in\Delta(A)}
  \left\|
    \sum_\theta w_\theta J_\theta(x)D(x)
  \right\|_\infty < 1,
\]
then \(F_\tau(\cdot;\mathcal{A})\) is a contraction in sup norm and has a
unique fixed point. A simple sufficient condition for
\(g_a(x_a)=c_a x_a^\rho\), \(\rho\ge 1\), is
\[
  \frac{\max_a c_a\rho}{2\tau}<1,
\]
because every masked-softmax Jacobian has \(\ell_\infty\) operator norm at most
\(1/(2\tau)\). If \(0<\rho<1\), the same statement holds on any region
\(x_a\ge\epsilon>0\) after replacing \(\max_a c_a\rho\) by
\(\max_a c_a\rho\epsilon^{\rho-1}\).
\end{proposition}

\begin{proof}
For a fixed awareness mask, unavailable actions have zero choice probability
and zero utility derivative, so the usual softmax Jacobian applies on the
feasible face of the simplex. The chain rule gives
\[
  \nabla_xF_\tau(x;\mathcal{A})
  =
  -\sum_\theta w_\theta J_\theta(x)D(x).
\]
The displayed operator bound implies
\(\|F_\tau(x;\mathcal{A})-F_\tau(y;\mathcal{A})\|_\infty
\le L\|x-y\|_\infty\) by the mean-value theorem on the simplex. Banach's
fixed-point theorem gives uniqueness. For a softmax row \(p\), the absolute
row sum of its utility Jacobian is
\(2p_a(1-p_a)/\tau\le 1/(2\tau)\), yielding the scalar sufficient condition.
The \(\epsilon\)-restricted nonlinear case uses
\(\max_a g'_a(x_a)\le\max_a c_a\rho\epsilon^{\rho-1}\).
\end{proof}

This condition is deliberately sufficient, not necessary. It covers the
fixed-awareness logit phase, including heterogeneous masks and action-specific
congestion slopes, but it does not rule out multiple fixed points when
temperature is low, congestion is steep, or awareness itself is still moving.
Those cases are handled experimentally as sensitivity regimes rather than
claimed as globally unique equilibria.

\begin{proposition}[Sensitivity to revelation size]
Under the contraction condition above, let \(\mathcal{A}'=R^r_S(\mathcal{A})\)
reveal action \(r\) to \(q=|S\setminus\{i:r\in A_i\}|\) newly aware agents. If
the direct effect of changing the masks of these agents is at most \(q/n\) in
the aggregate logit response, then the fixed points satisfy
\[
  \|x^\star(\mathcal{A}')-x^\star(\mathcal{A})\|_\infty
  \le
  \frac{q/n}{1-L}.
\]
In particular, the revealed action's equilibrium share can rise by at most this
amount.
\end{proposition}

\begin{proof}
Let \(F\) and \(F'\) be the two response maps. At the old fixed point
\(x^\star\), changing availability for \(q\) of \(n\) agents can change the
population-average response by at most \(q/n\). For the new fixed point
\(x'^\star\),
\[
\|x'^\star-x^\star\|_\infty
\le
\|F'(x'^\star)-F'(x^\star)\|_\infty
+\|F'(x^\star)-F(x^\star)\|_\infty
\le
L\|x'^\star-x^\star\|_\infty+q/n.
\]
Rearranging proves the bound.
\end{proof}

\begin{proposition}[Nonlinear price-of-information bound and tightness]
Let \(g_r(x)=c x^\rho\), and suppose revelation raises the equilibrium share of
the focal action from \(x\) to \(x+\delta\). The target price of information is
\[
  P_r=c\bigl((x+\delta)^\rho-x^\rho\bigr).
\]
For \(\rho\ge 1\),
\[
  c\rho x^{\rho-1}\delta
  \le
  P_r
  \le
  c\rho (x+\delta)^{\rho-1}\delta.
\]
Combining with the sensitivity proposition gives
\[
  P_r
  \le
  c\rho\left(x+\frac{q/n}{1-L}\right)^{\rho-1}
  \frac{q/n}{1-L}.
\]
The upper bound is tight in the zero-temperature limit when every newly aware
agent selects \(r\). If revelation causes no new adoption, then \(P_r=0\), so
visibility alone is not sufficient.
\end{proposition}

\begin{proof}
The exact expression follows by subtracting post- and pre-revelation congestion
payoffs. The two inequalities are the mean-value theorem applied to
\(\xi\mapsto c\xi^\rho\), whose derivative is increasing for \(\rho\ge 1\).
The sensitivity substitution gives the displayed bound. In the deterministic
limit, \(\delta=q/n\) whenever all newly aware agents choose the revealed action,
so the bound is attained up to the contraction slack; if \(\delta=0\), the
exact expression gives zero.
\end{proof}

\begin{proposition}[Dynamic awareness is not reducible to a static count game]
Assume \(\eta>0\) or at least one agent's feasible set changes under
revelation. There is no static congestion game whose payoffs and logit response
depend only on current action counts and fixed feasible action sets that
matches EHMG one-step behavior for all histories. In particular, two EHMG
histories can have the same realized count vector \(x\) but different awareness
profiles \(\mathcal{A}\neq\mathcal{A}'\), and hence different next-period logit
responses after a revelation operator.
\end{proposition}

\begin{proof}
A static count-based congestion game assigns the same utilities and the same
logit response to any two histories with the same current count vector. EHMG
does not: the perceived utility contains \(\eta(1-\pi_r)\), and the action mask
\(\mathbf{1}\{r\in A_i\}\) changes when revelation occurs. Thus two histories
with identical realized counts but different \(\pi_r\) or different agent masks
generate different response probabilities for \(r\). A fixed smaller action
space can match one awareness profile, but it cannot simultaneously match both
profiles and the transition between them.
\end{proof}

\begin{proposition}[Count-equivalent histories can have different next choices]
Let \(N=\{1,2\}\), \(A=\{b,r\}\), \(v_b=v_r=0\), \(c=0\),
\(\eta>0\), and temperature \(\tau>0\). Consider two histories with the same
current count vector \(n_b=2,n_r=0\). In history \(H\),
\(A_1=\{b,r\}\) and \(A_2=\{b\}\), so \(\pi_r=1/2\). In history \(H'\),
\(A_1=A_2=\{b,r\}\), so \(\pi_r=1\). Then the next-period logit probabilities
for action \(r\) differ:
\[
  \Pr_H(a_1=r)=
  \frac{\exp(\eta/(2\tau))}{1+\exp(\eta/(2\tau))},
  \qquad
  \Pr_H(a_2=r)=0,
\]
whereas
\[
  \Pr_{H'}(a_1=r)=\Pr_{H'}(a_2=r)=\frac{1}{2}.
\]
Thus a model observing only the current count vector cannot match EHMG
transition kernels across awareness profiles.
\end{proposition}

\begin{proof}
Both histories have the same realized counts and the same intrinsic action
values. They differ only in awareness. In \(H\), agent 1 receives scarcity
bonus \(\eta(1-\pi_r)=\eta/2\) for \(r\), while agent 2 cannot choose \(r\). In
\(H'\), both agents can choose both actions and the scarcity bonus is zero.
Substituting these utilities into the logit formula gives the stated
probabilities.
\end{proof}

\begin{proposition}[Fixed-awareness EHMG is a congestion potential game]
Fix awareness sets \(A_i\) and let payoffs be
\(u_i(a)=v_{a_i}-g_{a_i}(n_{a_i}(a))\), where each \(g_a\) depends only on the
number of agents choosing \(a\). Then the finite EHMG stage game is an exact
potential game on the restricted strategy spaces \(A_i\), with potential
\[
  \Phi(a)=\sum_{b\in A}\sum_{m=1}^{n_b(a)}\left(v_b-g_b(m)\right).
\]
\end{proposition}

\begin{proof}
Consider a unilateral deviation by agent \(i\) from \(x\) to \(y\). Only counts
on \(x\) and \(y\) change. The change in \(\Phi\) removes the marginal term
\(v_x-g_x(n_x)\) and adds \(v_y-g_y(n_y+1)\), exactly matching the deviator's
payoff change. Awareness only restricts which deviations are feasible; it does
not alter the potential identity.
\end{proof}

\begin{proposition}[Stationary logit revision under fixed awareness]
For fixed awareness sets and asynchronous logit revision with temperature
\(\tau>0\), the induced Markov chain over feasible action profiles has
stationary distribution
\[
  \mu_\tau(a)=\frac{\exp(\Phi(a)/\tau)}
  {\sum_{a'}\exp(\Phi(a')/\tau)}.
\]
Thus fixed-awareness EHMG inherits the standard regularized equilibrium
structure of finite potential games.
\end{proposition}

\begin{proof}
In an exact potential game, the logit transition ratio between two profiles
that differ by one agent equals \(\exp((\Phi(a')-\Phi(a))/\tau)\). This is the
detailed-balance condition for the Gibbs measure above. Since the finite chain
has positive probability on all feasible unilateral revisions, the stationary
distribution is unique on the feasible profile graph.
\end{proof}

\begin{proposition}[Visibility can destroy the payoff it reveals]
Suppose action \(a\) has realized payoff \(u_a(n_a)=v_a-c n_a/n\) with \(c>0\).
If a visibility shock increases its realized crowd from \(n_a\) to \(n'_a>n_a\)
while \(v_a\) and \(c\) are fixed, then the action payoff falls by
\[
  u_a(n_a)-u_a(n'_a)=c\frac{n'_a-n_a}{n}>0.
\]
\end{proposition}

\begin{proof}
Substitute the two crowd levels into the payoff function and subtract. The
strict inequality follows from \(c>0\) and \(n'_a>n_a\).
\end{proof}

\begin{proposition}[A welfare condition for harmful revelation]
Under linear congestion, suppose a visibility expansion moves \(m\) agents from
action \(b\) to a newly revealed action \(r\), with pre-shock counts \(n_b\) and
\(n_r\). Aggregate welfare decreases whenever
\[
  \frac{2c}{n}\bigl(n_r-n_b+m\bigr) > v_r-v_b.
\]
\end{proposition}

\begin{proof}
Total welfare is \(W=\sum_a n_a v_a-\frac{c}{n}\sum_a n_a^2\). The move changes
only actions \(b\) and \(r\). Hence
\[
  \Delta W
  =
  m(v_r-v_b)
  -
  \frac{c}{n}\left[(n_r+m)^2-n_r^2+(n_b-m)^2-n_b^2\right],
\]
which simplifies to
\[
  \Delta W
  =
  m(v_r-v_b)
  -
  \frac{2cm}{n}\bigl(n_r-n_b+m\bigr).
\]
For \(m>0\), \(\Delta W<0\) exactly under the stated condition.
\end{proof}

\begin{proposition}[Monotone response informational-loss condition]
Let a response rule assign a weakly larger probability to action \(r\) whenever
its perceived utility increases and all other perceived utilities are fixed. If
a public revelation of \(r\) weakly increases its perceived utility for at
least one newly aware agent and strictly increases expected adoption of \(r\),
then every strictly decreasing realized payoff \(u_r(n_r)\) yields a strictly
lower expected payoff for action \(r\) after revelation.
\end{proposition}

\begin{proof}
The response assumption implies that the post-revelation adoption count of
\(r\) first-order stochastically dominates the pre-revelation count and is
strictly larger with positive probability. Since \(u_r\) is strictly decreasing
in the count, taking expectations reverses the order.
\end{proof}

\begin{definition}[Price of information]
For focal action \(r\), awareness profile \(\mathcal{A}\), and revelation set
\(S\), define the target price of information as
\[
  P_r(S;\mathcal{A})
  =
  u_r(x_r(\mathcal{A}))
  -
  u_r(x_r(R^r_S(\mathcal{A}))),
\]
where \(x(\mathcal{A})\) is the fixed-awareness mean-field logit equilibrium.
Under linear congestion this is exactly
\[
  P_r(S;\mathcal{A})
  =
  c\left(x_r(R^r_S(\mathcal{A}))-x_r(\mathcal{A})\right).
\]
\end{definition}

\begin{proposition}[Target price and welfare can disagree]
Under linear congestion with unit population mass, let a disclosure move
\(\delta\) mass from action \(b\) to focal action \(r\). The target price is
\(P_r=c\delta\), while welfare changes by
\[
  \Delta W =
  \delta(v_r-v_b)
  -
  c\left[(x_r+\delta)^2-x_r^2+(x_b-\delta)^2-x_b^2\right].
\]
Consequently the target price of information is not an aggregate welfare
measure. With \(c=1\), \(x_r=0.1\), \(x_b=0.9\), \(\delta=0.2\), and
\(v_r-v_b=1\), one obtains \(P_r=0.2\) and \(\Delta W=0.44>0\). With
\(x_r=0.8\), \(x_b=0.2\), \(\delta=0.1\), and \(v_r=v_b\), one obtains
\(P_r=0.1\) and \(\Delta W=-0.14<0\). Finally, from the same baseline
\((x_r,x_b,x_d)=(0.1,0.6,0.3)\) with \(c=1\),
\((v_r,v_b,v_d)=(2.0,1.9,1.0)\), policy \(S\) moving \(0.1\) mass from \(b\)
to \(r\) has \(P_r(S)=0.1\) and \(\Delta W(S)=0.09\), while policy \(T\)
moving \(0.2\) mass from \(d\) to \(r\) has \(P_r(T)=0.2\) and
\(\Delta W(T)=0.20\). Minimizing target price prefers \(S\), but maximizing
aggregate welfare prefers \(T\).
\end{proposition}

\begin{proof}
The expression follows by subtracting pre- and post-disclosure welfare
\(\sum_a x_av_a-c\sum_a x_a^2\). The numerical claims are direct
substitutions. The examples show that target payoff loss can coexist with
welfare gain, welfare loss, or policy rankings opposite to welfare rankings.
\end{proof}

\begin{definition}[Budgeted welfare-safe disclosure]
Given candidate disclosure groups \(G_1,\ldots,G_m\), budget \(B\), welfare
floor \(\underline W\), and awareness profile \(\mathcal{A}\), choose
\(I\subseteq\{1,\ldots,m\}\) with \(|I|\le B\) to maximize newly aware agents
subject to
\[
  W(R^r_{\cup_{j\in I}G_j}(\mathcal{A}))\ge \underline W.
\]
\end{definition}

\begin{definition}[Minimum harmful revelation]
Given candidate groups and a target price threshold \(\gamma\), find a minimum
cardinality index set \(I\) such that
\[
  P_r(\cup_{j\in I}G_j;\mathcal{A})\ge \gamma.
\]
\end{definition}

\begin{proposition}[Minimum harmful revelation is NP-hard]
Minimum harmful revelation is NP-hard, even when all candidate groups reveal
the same focal action.
\end{proposition}

\begin{proof}
Reduce from set cover. Let the set-cover universe be
\(U=\{1,\ldots,n\}\) and let the candidate sets be
\(\mathcal{S}_1,\ldots,\mathcal{S}_m\). Build an EHMG instance with one agent
per element of \(U\), one default action \(b\), and one focal action \(r\).
Initially every agent is aware only of \(b\). Candidate disclosure group
\(G_j\) reveals \(r\) exactly to the agents corresponding to
\(\mathcal{S}_j\). Choose \(v_r-v_b\) large enough and temperature small enough
that every newly aware agent strictly selects \(r\) in the deterministic
best-response limit, while unaware agents cannot select \(r\). With
\(g_r(x)=c x^\rho\), revealing a family \(I\) therefore makes the equilibrium
share of \(r\) equal to
\[
  x_r(I)=\frac{|\cup_{j\in I}\mathcal{S}_j|}{n}.
\]
Set the harmful threshold to
\[
  \gamma=c\left(1^\rho-0^\rho\right)=c.
\]
Then \(P_r(I)\ge\gamma\) if and only if
\(\cup_{j\in I}\mathcal{S}_j=U\). Thus a minimum harmful revelation is exactly a
minimum set cover. The construction is polynomial. The same argument with
\(\gamma=c(k/n)^\rho\) gives partial-cover hardness for a \(k\)-agent harmful
threshold. For positive but sufficiently small logit temperature, strict utility
separation preserves the reduction by continuity with a fixed margin.
\end{proof}

\begin{proposition}[A normalized Braess visibility loss]
Consider a unit demand network with edges \(s\to x\) of cost \(f\), \(x\to t\)
of cost \(1\), \(s\to y\) of cost \(1\), and \(y\to t\) of cost \(f\). Without
edge \(x\to y\), the Wardrop equilibrium splits flow equally and has cost
\(1.5\). If the zero-cost edge \(x\to y\) is visible and usable, the Wardrop
equilibrium sends all flow along \(s\to x\to y\to t\) and has cost \(2.0\).
\end{proposition}

\begin{proof}
Without the middle edge, the two route costs are \(x_1+1\) and \(1+x_2\), with
\(x_1+x_2=1\). Wardrop equilibrium equalizes them at \(x_1=x_2=0.5\), yielding
cost \(1.5\). With the middle edge, the zig-zag route has cost
\((x_1+x_3)+(x_2+x_3)\), where \(x_3\) is zig-zag flow. At \(x_3=1\) and
\(x_1=x_2=0\), all available routes have cost at least \(2\), and the used route
has cost \(2\). No user can improve by deviating, so this is a Wardrop
equilibrium with higher cost.
\end{proof}

\begin{proposition}[EHMG revelation strictly raises information-constrained Wardrop cost]
Embed the normalized Braess network as an EHMG in which routes are actions and
route awareness is the feasible action set. Initially all agents are aware only
of the two outer routes, so the information-constrained Wardrop equilibrium has
cost \(1.5\). A public revelation operator that makes the zero-cost middle edge
common knowledge expands every feasible route set to include the zig-zag route.
The resulting Wardrop equilibrium has cost \(2.0\). Hence EHMG public revelation
strictly increases equilibrium congestion cost by \(0.5\).
\end{proposition}

\begin{proof}
The initial awareness profile restricts every agent to the two-route subgame,
whose Wardrop equilibrium is the equal split from the previous proposition and
has cost \(1.5\). Public revelation is \(R^r_N\) for the zig-zag route \(r\):
every agent can now choose it and every agent observes that all others can
choose it. The full route set is exactly the Braess network with the middle
edge available, whose Wardrop equilibrium sends all flow through the zig-zag
route at cost \(2.0\). The same physical network and demand are used in both
states; only the awareness-transition operator changes the equilibrium.
\end{proof}

\section{Experiments}

All experiments are generated by \texttt{scripts/materialize\_experiments.py}.
The script writes raw and aggregate CSV files, selected manuscript figures,
auxiliary plot previews, an experiment JSON audit, and a reproducibility
manifest. The source package includes code and result artifacts.

\paragraph{E1: awareness sweep.}
We vary initial awareness of a hidden third action from \(0.03\) to \(1.00\).
Each condition uses 24 deterministic-seed replicates of a 501-agent, 280-round
game. The outcome is the hidden action's late-window payoff and choice rate.

\paragraph{E2: public visibility shock.}
The hidden action begins with \(8\%\) awareness. At round 180, it is revealed to
all agents. We compare desire biases \(0.00\), \(0.20\), and \(0.45\), using 18
replicates per condition. The primary outcomes are hidden-action payoff
collapse, hidden-action choice-rate jump, and welfare change.

\paragraph{E3: horizon--desire grid.}
We cross horizons \(h\in\{1,3,8,20\}\) with desire biases
\(\lambda\in\{0,0.12,0.28,0.45\}\). All agents know all actions, so this
experiment isolates bounded horizon and desire from awareness. The outcome is
late-window volatility of the focal action's choice rate.

\paragraph{E4: Braess information comparison.}
We compute the analytic normalized Braess costs from the Braess visibility
proposition and write the restricted and full-information route flows to CSV.

\paragraph{E5: transferred algorithmic mechanism layer.}
We materialize the general-purpose code layer behind the simulations in
formal mechanism language. The layer includes binary oracle identification
bounds, a k-Clique to compatible-coalition reduction, min-fill mechanism graph
kernelization, bounded-treewidth dynamic-programming work estimates,
risk-plus-ambiguity policy scoring, message-passing repair for binary belief
consistency checks, and a reduced strategic SDE over salience, dispersion, and
delay pressure.

\paragraph{E6: robustness phase diagram.}
To separate structural behavior from parameter-induced behavior, we repeat the
visibility shock with no desire bias and no direct desire term. We vary
temperature \(\tau\in\{0.08,0.16,0.32,0.64\}\), congestion curvature
\(\rho\in\{0.75,1,1.5,2\}\), and learning rule
\(\{\)softmax, epsilon-greedy\(\}\). The outcome is the post-revelation payoff
collapse of the revealed action and welfare change.

\paragraph{E7: disclosure and desire-as-attention.}
We compare no revelation, private partial revelation, correlated partial
disclosure, public common revelation, desire acting through attention/awareness,
and desire acting through utility. This addresses the behavioral-identification
concern: desire can influence which action is represented, not only how a
represented action is valued.

\paragraph{E8: counterfactual baselines.}
We compare the public EHMG shock with ordinary full-awareness congestion,
full-awareness minority dynamics without epistemic scarcity, private partial
revelation, and public revelation without the scarcity bonus. This tests whether
the EHMG formalism predicts behavior absent from standard baselines.

\paragraph{E9: disclosure optimization.}
We instantiate budgeted welfare-safe disclosure and minimum harmful revelation
on a six-group disclosure instance. Exact enumeration, greedy selection, full
public revelation, and the minimum harmful subset are all evaluated under the
same mean-field logit outcome map.

\paragraph{E10: strategic trajectory imaging.}
We encode the public-visibility-shock trace as a line image,
Gramian-angular hidden-rate image, recurrence image, Gramian payoff image,
invertible spectral carrier, and phase-scrambled control.

\paragraph{E11: trace-image recognition and stress benchmark.}
We generate four labeled strategic regimes--ordinary congestion, epistemic
scarcity, public revelation, and attention-mediated discovery--with deterministic
train/test splits over seeds. A nearest-centroid classifier is evaluated on raw
trace features, Fourier features, recurrence-image features, and the full image
bundle. We then run a harder binary stress benchmark: every awareness-cascade
trace is paired with an iterative amplitude-adjusted Fourier-transform (IAAFT)
null that preserves the trace's sorted values and approximately preserves its
Fourier magnitudes. The train and test splits differ in reveal time,
temperature, congestion curvature, horizon distribution, action-value gap,
circular phase, and noise. Raw positions, Fourier magnitudes, a ROCKET-style
random-convolution time-series baseline, recurrence-image summaries, an
image-structure bundle, and a small local-filter baseline trained from scratch
are evaluated. Additional CSV checks vary image resolution, recurrence
threshold, Gramian variant, spectral-carrier construction, and phase scrambling;
metadata-only and marginal-only leakage probes test seed, trace-length,
reveal-time, parameter, and marginal-statistic shortcuts.

\paragraph{E12: external-trace adapter.}
No real external trace dataset is included. Instead, the code provides a
long-format CSV adapter for future route-choice, online-attention,
experimental-minority-game, or gambling-prediction traces, plus a synthetic
external generator with cascade and cyclic dynamics that do not call the EHMG
simulator. The canonical required columns are \texttt{trace\_id},
\texttt{time}, \texttt{signal\_name}, and \texttt{value}; optional columns
record label, group, and metadata. The reproducibility script writes both a
long-format adapter input and a smoke-test summary with
\(\ExternalAdapterTraceCount\) traces.

\section{Results}

The awareness sweep is reported textually rather than as a manuscript figure.
The hidden-action payoff falls with fitted slope \(\AwarenessSlope\), moving
from \(\LowestAwarenessPayoff\) at the lowest awareness level to
\(\FullAwarenessPayoff\) under full awareness. This supports the inverse
abundance claim in its simplest form: the hidden action is not valuable merely
because \(v_C>v_A=v_B\); it is valuable because too few agents can select it.
When awareness approaches one, its crowd increases and its advantage decays.

\begin{figure}[H]
  \centering
  \includegraphics[width=0.82\linewidth]{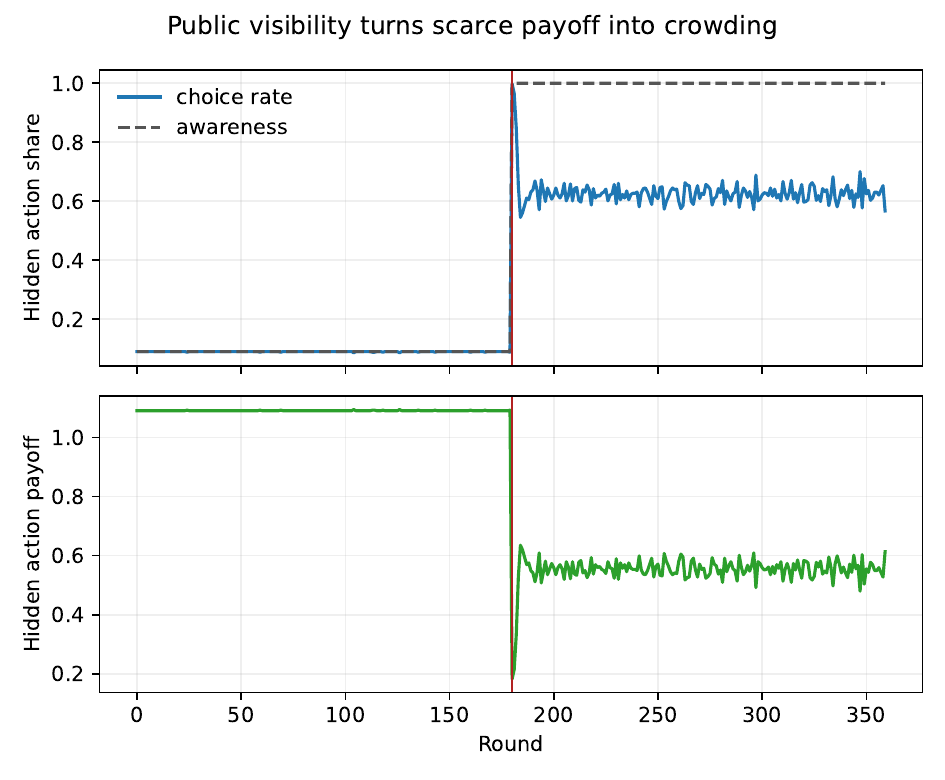}
  \caption{A public revelation shock makes the hidden action common and
  crowded. In the strong desire condition, the hidden-action choice rate rises by
  \(\StrongVisibilityChoiceJump\), while hidden-action payoff falls by
  \(\StrongVisibilityCollapse\).}
  \label{fig:shock}
\end{figure}

Figure~\ref{fig:shock} shows the mechanism dynamically. Before revelation, the
hidden action has low awareness and high realized payoff. After revelation, the
same action is selected by more agents. Desire accelerates the movement toward
the revealed action, converting epistemic scarcity into crowding.

The horizon--desire grid isolates bounded observation and desire without using
another figure: focal-action volatility ranges from \(\MinHorizonVolatility\)
to \(\MaxHorizonVolatility\), and shorter horizons with stronger desire produce
larger oscillations. The normalized Braess comparison is also reported directly:
route-set expansion raises equilibrium cost from \(\BraessRestrictedCost\) to
\(\BraessFullCost\), an increase of \(\BraessCostIncrease\).

The reduced strategic dynamics and algorithmic-transfer checks are retained as
auditable generated artifacts rather than manuscript figures. Their CSV record
reports
\(\OracleBoundFiveHundredOne\) oracle queries for 501 strategic states,
\(\FanoObservationBound\) Fano-limited observations in an eight-class task, a
treewidth reduction from \(\TreewidthBefore\) to \(\TreewidthAfter\), policy
surprise \(\PolicySurprise\), policy expected success
\(\PolicyExpectedSuccess\), and \(\BeliefRepairInconsistency\) remaining
belief-consistency violations.

The robustness and disclosure-mode checks are likewise kept textual. With
desire removed, the revealed action's payoff collapse remains positive in
\(\RobustnessPositiveCollapse\) of raw grid cells, with raw collapses ranging
from \(\RobustnessMinCollapse\) to \(\RobustnessMaxCollapse\). Public common
revelation produces payoff collapse \(\PublicDisclosureCollapse\);
desire-as-attention changes hidden-action choice rate by
\(\AttentionDesireChoiceJump\), while direct desire-as-utility changes it by
\(\UtilityDesireChoiceJump\). These comparisons keep the behavioral
identification claim without adding near-duplicate charts.

The counterfactual baselines are also reported without a separate bar chart.
The EHMG public revelation condition produces hidden-action payoff collapse
\(\CounterfactualEhmgCollapse\), while the ordinary full-awareness congestion
baseline changes by only \(\CounterfactualOrdinaryCollapse\), giving epistemic
gap \(\CounterfactualEpistemicGap\). This addresses the main experimental
criticism: the mechanism is not merely ``more users of a congestible action
lowers its payoff.'' In the ordinary full-awareness baseline there is no hidden
action to reveal, so the same pre/post measurement is essentially flat. Public
revelation without the scarcity bonus still collapses payoff by
\(\CounterfactualNoBonusCollapse\), showing that action-space expansion itself
is doing work.

\paragraph{Formal audit experiments.}
Four CSV-backed audit experiments make the theorem claims easier to inspect.
The non-reducibility witness holds current action counts fixed but changes
awareness masks, producing transition-kernel total variation
\(\AuditNonreducibilityTv\) and aggregate revealed-action probability gap
\(\AuditNonreducibilityAggregateGap\). The masked-logit audit evaluates
\(\AuditContractionCount\) finite cases and records
\(\AuditContractionConservativeCases\) cases where the sufficient scalar bound
fails while the measured local Jacobian remains contractive; the largest
measured local norm is \(\AuditContractionEmpiricalMaxNorm\). The price/welfare
grid includes \(\AuditPriceWelfareUpCases\) welfare-increasing and
\(\AuditPriceWelfareDownCases\) welfare-decreasing target-price cases, with
\(\AuditPriceWelfareOppositePairs\) policy pair where minimizing target price and
maximizing welfare disagree. The same-reach signal example holds recipient
reach fixed while posterior awareness beliefs vary, yielding adoption spread
\(\AuditSignalAdoptionSpread\) and relative-odds range
\([\AuditSignalMinRelativeOdds,\AuditSignalMaxRelativeOdds]\). These audits are
diagnostic checks on the formal distinctions, not substitutes for the proofs.

\begin{figure}[H]
  \centering
  \includegraphics[width=0.82\linewidth]{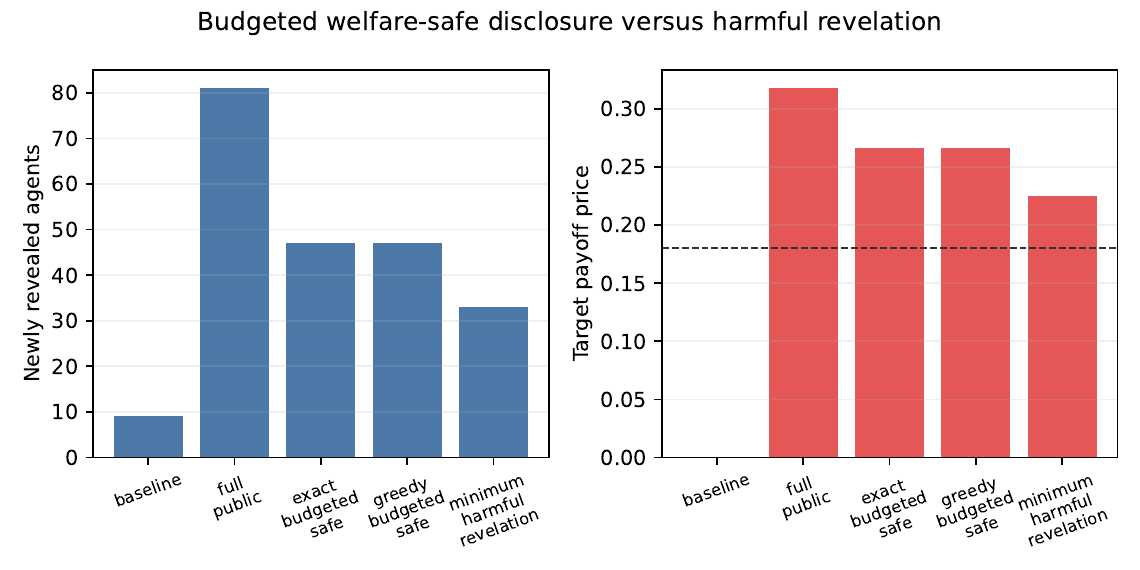}
  \caption{Algorithmic disclosure tasks. Exact budgeted welfare-safe disclosure
  reveals \(\DisclosureExactRevealed\) agents, the greedy policy reveals
  \(\DisclosureGreedyRevealed\), and the greedy gap is
  \(\DisclosureGreedyGap\). Full public revelation has target price
  \(\DisclosurePublicPrice\), while minimum harmful revelation crosses the
  target-price threshold with \(\DisclosureMinHarmfulGroups\) groups and price
  \(\DisclosureMinHarmfulLoss\).}
  \label{fig:optimization}
\end{figure}

\begin{figure}[H]
  \centering
  \includegraphics[width=0.86\linewidth]{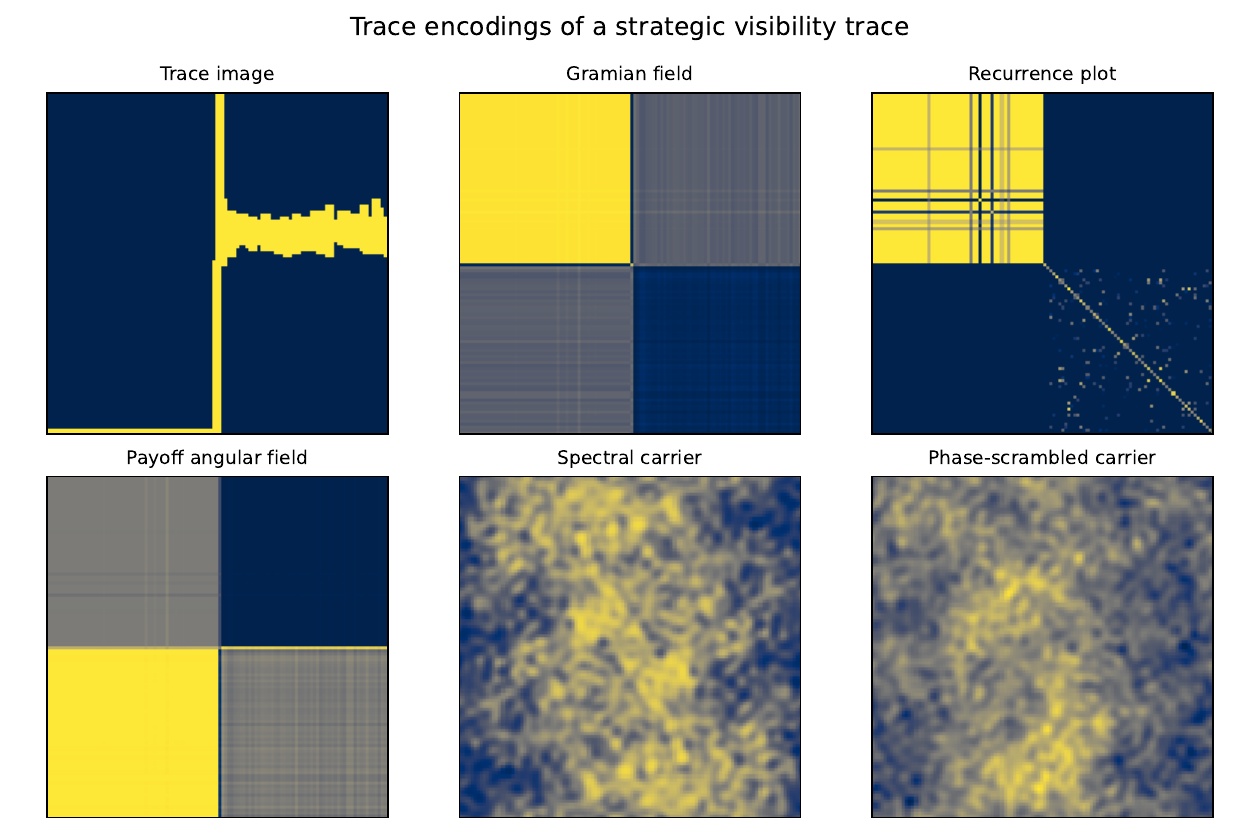}
  \caption{Trace encodings of one strategic trajectory. The six-panel
  diagnostic includes a trace image, Gramian and recurrence images, a payoff
  angular field, an invertible spectral carrier, and a phase-scrambled carrier.
  The spectral carrier reconstructs the z-scored hidden-rate trace at
  \(\TraceSpectralPsnr\) dB after 8-bit quantization; the phase control changes
  local structure by \(\TracePhaseStructureDelta\) while leaving the normalized
  radial power profile close, with relative error \(\TracePowerProfileRelError\).}
  \label{fig:trace}
\end{figure}

Figure~\ref{fig:trace} shows the image representations used by the trace
recognition benchmark. The figure is not the evidence by itself; it records what
the feature pipelines see before the held-out detection task.

The basic four-regime recognition task is reported without a separate bar
chart: over \(\TraceDetectionTestExamples\) held-out traces, raw and Fourier
baselines reach \(\TraceDetectionRawAccuracy\) and
\(\TraceDetectionFourierAccuracy\), recurrence-image features reach
\(\TraceDetectionRecurrenceAccuracy\), and the image bundle reaches
\(\TraceDetectionImageAccuracy\). Because this task is
partially separable from raw trajectories, the stricter matched-null benchmark
is reported textually rather than as another bar chart. The matched null removes
the easy marginal and spectral cues, and the test split changes the event
location and game parameters. Raw trace features reach
\(\TraceStressRawAccuracy\), Fourier magnitudes reach
\(\TraceStressFourierAccuracy\), the ROCKET-style time-series baseline reaches
\(\TraceStressRocketAccuracy\), recurrence-image summaries reach
\(\TraceStressRecurrenceAccuracy\), and the image-structure bundle reaches
\(\TraceStressImageAccuracy\). The small local-filter model reaches
\(\TraceStressConvAccuracy\). The image-encoding gain over the best non-image
baseline, including ROCKET, is \(\TraceStressImageGain\).
The ablation table reports \(\TraceAblationCount\) image-encoding controls, and
the metadata-only leakage probe obtains \(\TraceLeakageMetadataAccuracy\), consistent
with the paired cascade/null construction rather than a seed or parameter
shortcut.

The remaining generated quantities are in the CSV tables and JSON experiment
summary; the manuscript avoids a large catch-all summary table.

\section{Interpretation}

\paragraph{Gambling.}
In a pure lottery, previous independent outcomes do not change the next draw.
Yet gambling decisions are rarely pure computations of objective probability.
They are threshold decisions under arousal, loss framing, pattern search, and
desire. EHMG represents this by separating realized payoff from perceived
utility. Desire does not need to rewrite the true probability; it only needs to
shift a borderline action into the chosen set.

\paragraph{The forest horizon.}
An agent in a forest can feel locally complete while being globally uninformed.
This is not irrationality by itself. It is horizon-bounded rationality: the
agent acts coherently inside the events represented in the local model. EHMG
formalizes this as \(A_i(t)\) and \(h_i\). The world can contain actions and
states that the agent neither observes nor evaluates.

\paragraph{Unobservable domains.}
Claims about domains outside observation should not be modeled as if all agents
share one exhaustive state space and merely disagree about probabilities. In
many cases, agents fill different state spaces. EHMG does not adjudicate
metaphysical truth. It provides a formal language for the fact that
agents can coordinate, disagree, or fragment because their horizons and
awareness sets differ.

\paragraph{Formal epistemic status.}
The present model is a lattice model of action awareness, not a full epistemic
logic with hierarchies of subjective state spaces. Public revelation is modeled
as a lattice jump that is common to the population; private and correlated
revelations are different transition operators on the same lattice. This is
deliberately weaker than a complete game-with-unawareness foundation, but
stronger than treating unawareness as a static omitted action.

\paragraph{Strategic value.}
The slogan ``what is many is little; what is little is many'' becomes precise:
an action's strategic value is a function of intrinsic value, crowding, and
visibility. When a hidden action is rare in awareness, it can yield high payoff.
When it becomes common knowledge, it can lose the very advantage that made it
attractive.

\section{Limitations}

The model is deliberately minimal. It does not estimate parameters from human
subjects, does not claim that desire is the only source of gambling bias, and
does not replace equilibrium analysis in network games. It also does not give a
complete epistemic-logic treatment of awareness hierarchies; the contribution is
the awareness-transition game and disclosure-optimization layer. The simulations
are agent-based stress tests for a proposed mechanism. The trace-encoding
experiments establish trace-to-image structure under controlled synthetic game
traces; they include a lightweight local-filter baseline but do not yet
evaluate large image models, real video, or human-generated strategic
trajectories. Future work should fit the model to experimental minority-game
play, online attention cascades, route choice data, or gambling tasks in which
predictions and likelihood judgments are separately elicited, and should test
whether the matched-null trace-encoding result transfers through the
external-trace adapter.

\section{Conclusion}

Epistemic horizon minority games turn a philosophical intuition into a testable
formal object. Agents choose inside bounded horizons, awareness evolves
on a finite lattice, some actions are valuable because few agents see them,
desire can push marginal choices toward crowding, and common visibility can
erase strategic payoff. The added disclosure tasks make the mechanism
algorithmic: one can ask which revelation policy is safe, which small
revelation is harmful, and which traces reveal the underlying regime. The result
is a compact formal and computational account of why, in strategic
environments, abundance can reduce value and scarcity can expand it.

\FloatBarrier
\bibliographystyle{plainnat}
\bibliography{refs}

\end{document}

%% file: experiment_macros.tex
\newcommand{\AwarenessSlope}{-0.418}
\newcommand{\LowestAwarenessPayoff}{1.152}
\newcommand{\FullAwarenessPayoff}{0.755}
\newcommand{\StrongVisibilityCollapse}{0.545}
\newcommand{\StrongVisibilityChoiceJump}{0.545}
\newcommand{\MaxHorizonVolatility}{0.495}
\newcommand{\MinHorizonVolatility}{0.022}
\newcommand{\BraessRestrictedCost}{1.500}
\newcommand{\BraessFullCost}{2.000}
\newcommand{\BraessCostIncrease}{0.500}
\newcommand{\OracleBoundFiveHundredOne}{9}
\newcommand{\FanoObservationBound}{6}
\newcommand{\TreewidthBefore}{5}
\newcommand{\TreewidthAfter}{2}
\newcommand{\PolicySurprise}{0.664}
\newcommand{\PolicyExpectedSuccess}{0.709}
\newcommand{\BeliefRepairInconsistency}{0}
\newcommand{\RobustnessPositiveCollapse}{1.000}
\newcommand{\RobustnessMinCollapse}{0.141}
\newcommand{\RobustnessMaxCollapse}{0.406}
\newcommand{\PublicDisclosureCollapse}{0.333}
\newcommand{\AttentionDesireChoiceJump}{0.007}
\newcommand{\UtilityDesireChoiceJump}{0.550}

\newcommand{\TraceSpectralPsnr}{59.9}
\newcommand{\TracePhaseStructureDelta}{0.034}
\newcommand{\TracePowerProfileRelError}{0.000}
\newcommand{\CounterfactualEhmgCollapse}{0.340}
\newcommand{\CounterfactualOrdinaryCollapse}{0.000}
\newcommand{\CounterfactualNoBonusCollapse}{0.335}
\newcommand{\CounterfactualEpistemicGap}{0.339}
\newcommand{\DisclosureExactRevealed}{47}
\newcommand{\DisclosureGreedyRevealed}{47}
\newcommand{\DisclosureGreedyGap}{0}
\newcommand{\DisclosurePublicPrice}{0.318}
\newcommand{\DisclosureMinHarmfulGroups}{2}
\newcommand{\DisclosureMinHarmfulLoss}{0.225}
\newcommand{\TraceDetectionRawAccuracy}{1.000}
\newcommand{\TraceDetectionFourierAccuracy}{1.000}
\newcommand{\TraceDetectionRecurrenceAccuracy}{1.000}
\newcommand{\TraceDetectionImageAccuracy}{1.000}

\newcommand{\TraceDetectionTestExamples}{40}
\newcommand{\TraceStressRawAccuracy}{0.542}
\newcommand{\TraceStressFourierAccuracy}{0.479}
\newcommand{\TraceStressRocketAccuracy}{0.583}
\newcommand{\TraceStressRecurrenceAccuracy}{0.604}
\newcommand{\TraceStressImageAccuracy}{0.958}
\newcommand{\TraceStressConvAccuracy}{0.604}

\newcommand{\TraceStressImageGain}{0.375}

\newcommand{\TraceLeakageMetadataAccuracy}{0.500}
\newcommand{\TraceAblationCount}{7}
\newcommand{\ExternalAdapterTraceCount}{16}
\newcommand{\AuditNonreducibilityTv}{0.250}
\newcommand{\AuditNonreducibilityAggregateGap}{0.134}
\newcommand{\AuditContractionCount}{48}
\newcommand{\AuditContractionConservativeCases}{10}
\newcommand{\AuditContractionEmpiricalMaxNorm}{4.067}
\newcommand{\AuditPriceWelfareUpCases}{3}
\newcommand{\AuditPriceWelfareDownCases}{1}
\newcommand{\AuditPriceWelfareOppositePairs}{1}
\newcommand{\AuditSignalAdoptionSpread}{0.284}
\newcommand{\AuditSignalMaxRelativeOdds}{1.822}
\newcommand{\AuditSignalMinRelativeOdds}{0.407}